\documentclass[11pt]{article}
\usepackage[margin=1in]{geometry}
\usepackage{setspace}
\usepackage{graphicx}
\usepackage{dsfont}
\usepackage{enumerate}
\usepackage{mathtools}
\usepackage[colorlinks,citecolor=blue]{hyperref}
\usepackage{amsmath}
\usepackage{amssymb}
\usepackage{nicefrac}
\usepackage{url}
\usepackage{complexity}
\usepackage{xcolor}
\usepackage[normalem]{ulem}

\usepackage[backend=biber, style=alphabetic, backref=true, hyperref=true, maxbibnames=99]{biblatex}
\usepackage[babel,english=british]{csquotes}
\DefineBibliographyStrings{english}{%
    backrefpage  = {cited on p.}, 
    backrefpages = {cited on pp.} 
}
\newcommand{\IN}{\mathbb{N}}

\newcommand{\VC}{\operatorname{VCdim}}

\newcommand{\stkout}[1]{\ifmmode\text{\sout{\ensuremath{#1}}}\else\sout{#1}\fi}
\newif\ifverbose
\verbosetrue

\usepackage{amsthm}
\newtheorem{theorem}{Theorem}[section]

\newtheorem{proposition}[theorem]{Proposition}

\newtheorem{corollary}[theorem]{Corollary}

\newtheorem{lemma}[theorem]{Lemma}
\newtheorem{definition}[theorem]{Definition}

\theoremstyle{definition}

\newtheorem{remark}[theorem]{Remark}

\numberwithin{equation}{section}

\begin{document}

\title{From Undecidability of Non-Triviality and Finiteness to Undecidability of Learnability}
\author{
Matthias C.~Caro\footnote{Part of this work was done at the Mathematics Department of the Technical University of Munich and at the Munich Center for Quantum Science and Technology (MCQST).} \thanks{\emph{ORCiD:} \href{https://orcid.org/0000-0001-9009-2372}{0000-0001-9009-2372}}\\
\small Institute for Quantum Information and Matter, Caltech, Pasadena, CA, USA\\
\small Dahlem Center for Complex Quantum Systems, Freie Universität Berlin, Berlin, Germany\\
\small \texttt{mcaro@caltech.edu}
}   
\date{}

\maketitle

\begin{abstract}
Machine learning researchers and practitioners steadily enlarge the multitude of successful learning models. They achieve this through in-depth theoretical analyses and experiential heuristics.
However, there is no known general-purpose procedure for rigorously evaluating whether newly proposed models indeed successfully learn from data.

We show that such a procedure cannot exist. For PAC binary classification, uniform and universal online learning, and exact learning through teacher-learner interactions, learnability is in general undecidable, both in the sense of independence of the axioms in a formal system and in the sense of uncomputability.
Our proofs proceed via computable constructions that encode the consistency problem for formal systems and the halting problem for Turing machines into whether certain function classes are trivial/finite or highly complex, which we then relate to whether these classes are learnable via established characterizations of learnability through complexity measures.
Our work shows that undecidability appears in the theoretical foundations of artificial intelligence: There is no one-size-fits-all algorithm for deciding whether a machine learning model can be successful. We cannot in general automatize the process of assessing new learning models.

\end{abstract}

\newpage
\section{Introduction}\label{SctIntroduction}

One of the foundational questions in machine learning theory is ``When is learning possible?'' This is the question for necessary and sufficient conditions for learnability. Such conditions have been identified for different learning models. They can take the form of requiring a certain, often combinatorial, complexity measure to be finite. Well-known examples of such complexity measures include the VC-dimension for binary classification in the PAC model, the Littlestone dimension for online learning, and different notions of teaching dimensions for teacher-learner interactions.

We consider a question that is slightly different from, but arguably just as important as the one above. Namely, we ask ``Can we decide whether learning is possible?'' At first glance, the ability to answer the first question might also seem to allow to resolve this second one. If, e.g., you know a complexity measure whose finiteness is equivalent to learnability, that gives you a criterion to decide learnability. However, whether this is indeed a satisfactory criterion strongly depends on the exact meaning of ``decide'' in the second question.

We consider two such meanings and thereby obtain two variants of the second question. The first is natural from a mathematician's perspective, namely ``If a class is learnable, can we prove that this is the case?'' The second is intimately familiar to computer scientists, namely ``Does there exist an algorithm that decides learnability?'' After specifying in either of these two ways what it means to ``decide whether learning is possible,'' we see that the answer to the second question is not trivially positive. 
Even given the definition of a complexity parameter that is finite if and only if learning is possible, answering the second question still requires a proof of finiteness of that complexity measure or an algorithm that decides whether the complexity measure is finite or not.

In fact, we show that the answer to the question ``Can we decide whether learning is possible?'' is, in general, negative for both of the variants introduced above and for different learning scenarios. In particular, we demonstrate this for learning models in which criteria for learnability in terms of complexity measures are known. More concretely, we consider binary classification, uniform and universal online learning, and the task of exactly identifying a function through teacher-learner interactions. We show in all these scenarios: On the one hand, there is a function class that is learnable but whose learnability cannot be proved. On the other hand, there is no general-purpose algorithm that, upon input of a class, decides whether it is learnable.
In our reasoning, we trace the impossibility of deciding learnability back to the folklore impossibility of deciding non-triviality and finiteness. This allows us to handle all the aforementioned different learning scenarios with a unified construction.

\subsection{Overview of the Results}\label{SbSctOverview}

Our undecidability results come in two flavours, one about provability in a formal system, the other about computability via Turing machines. We summarize our line of reasoning in Figure \ref{fig:flowchart} and explain it in more detail in the following paragraphs.

\begin{figure}[ht]
    \centering        
    \makebox[\textwidth][c]{
    \includegraphics{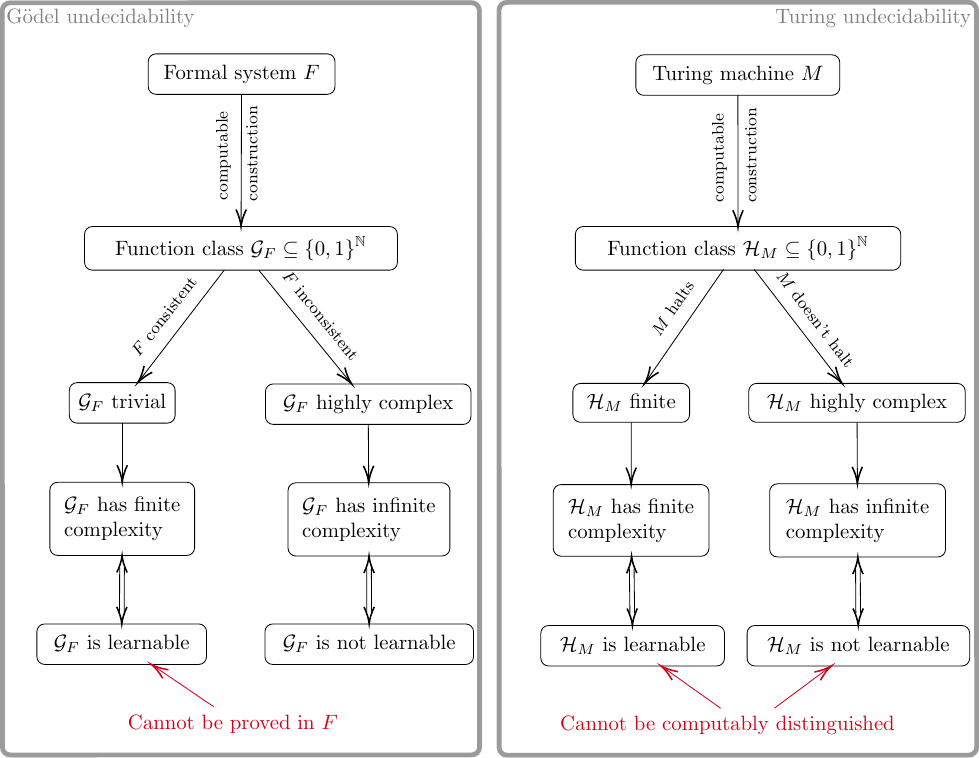}
}
    \caption{A depiction of our line of reasoning. ``Complexity'' is to be understood in terms of VC-dimension, teaching dimension, Littlestone dimension, or Littlestone trees, depending on the learning model. To conclude undecidability, we use Gödel's second incompleteness theorem and the uncomputability of the halting problem, respectively.}
    \label{fig:flowchart}
\end{figure}

Our general construction and proof strategy is the same for all the different learning scenarios considered in this work. On the one hand, given a recursively enumerable formal system $F$, we define a class $\mathcal{G}_F\subseteq\{0,1\}^\mathbb{N}$ (Definition \ref{DffGoedelFctClass}) that is trivial if $F$ is consistent, and highly complex if $F$ is inconsistent. If Gödel's second incompleteness theorem applies to $F$, we conclude that the function class $\mathcal{G}_F$ is trivial, but its triviality cannot be proved in $F$. On the other hand, given a Turing machine $M$, we define a class $\mathcal{H}_M\subseteq\{0,1\}^\mathbb{N}$ (Definition \ref{DffTuringFctClass}) that is finite if $M$ halts on the empty input, and highly complex, namely equal to $c_c(\{0,1\})$, if $M$ does not halt on the empty input. (Here, $c_c(\{0,1\})$ denotes the recursively enumerable set of $\{0,1\}$-valued functions with finite support in $\IN=\{0,1,2,\ldots\}$.)
Using a reduction to the halting problem, we use this to conclude that there is no general-purpose algorithm for deciding whether a computable binary-valued function class is finite or equal to $c_c(\{0,1\})$.

Our constructions start from $c_c(\{0,1\})$. Depending on the underlying object, i.e., the formal system or the Turing machine, we then further restrict the function class. We implement these restrictions based on consistency of finitely many provable theorems and halting after finitely many steps. Thereby, we give a constructive recipe for evaluating any desired element of the function class on any given input, which ensures that the classes are computable from the underlying object. For the Gödel scenario, this translates the assumption of the existence of a recursive enumeration of the provable theorems into a property of the function class. For the Turing scenario, this computability is necessary for a reduction to the halting problem. Abstracting from our recipe, we obtain a formalization of how an algorithm can be given a function class as input by identifying a computable function class (compare Definition \ref{DffComputableFunctionClass}) with its code w.r.t.~a universal Turing machine.

Combining our observations about triviality and finiteness of our constructed function classes with characterizations of learnability through complexity measures now leads us to corresponding statements about undecidability of learnability in different scenarios. 
We first study binary classification in \emph{Probably Approximately Correct (PAC)} learning. The relevant complexity measure for this learning scenario is the VC-dimension due to \cite{Vapnik.1971}.
By showing that the VC-dimension of $\mathcal{G}_F$ is finite if $\mathcal{G}_F$ is trivial and infinite if $\mathcal{G}_F$ is highly complex, we conclude: If Gödel's second incompleteness theorem applies to $F$, the function class $\mathcal{G}_F$ is PAC learnable, but its PAC learnability cannot be proved in $F$.
Similarly, as the VC-dimension of a finite class is finite and the VC-dimension of $c_c(\{0,1\})$ is infinite, we see that there is no computable way of deciding whether a computable $\{0,1\}$-valued function class has finite VC-dimension.

We apply a similar reasoning for learning scenarios with teacher-learner interactions, where a teacher can provide examples of a target function to help a learner identify that function. Here, the basic complexity measure is the teaching dimension \cite{Goldman.1995}, which is finite if and only if the teaching problem can be solved with finitely many examples. Combining this with our constructions of $\mathcal{G}_F$ and $\mathcal{H}_M$ allows us to show both Gödel and Turing undecidability of finiteness of the teaching dimension and thus learnability/teachability. We demonstrate the undecidability of one more decision problem motivated by teacher-learner interactions. Namely, in general, one cannot decide -- both in the sense of independence of the axioms of a formal system and in the sense of uncomputability -- whether a given function in a known class can be taught/learned from finitely many examples.

Finally, our constructions also yield undecidability results for uniform and universal online learning. For online learning with uniform mistake bounds, the Littlestone dimension \cite{Littlestone.1988} is the corresponding complexity parameter. For universal online learning, the relevant complexity condition is whether there exists an infinite Littlestone tree \cite{Bousquet.2020}. After showing that whether these complexity conditions are satisfied by $\mathcal{G}_F$ and $\mathcal{H}_M$ is again determined by whether $F$ is consistent and whether $M$ halts on the empty input, respectively, we conclude: Both uniform and universal online learnability are, in general, both Gödel and Turing undecidable.

Compared to prior work on undecidability in learning theory, which we review in Subsection \ref{SbSct:related-work}, our approach is at the same time more direct and is the first that simultaneously proves undecidability results for multiple established learning models both in the sense of formal independence and in the sense of uncomputability. Our main contribution consists in carefully elaborating the computational model, which then facilitates the construction and study of the function classes $\mathcal{G}_F$ and $\mathcal{H}_M$. As part of our model we formalize what it means for an algorithm to be given a function class as input, for which learnability is then to be decided. Conceptually, we show that many of the established learnability criteria in terms of complexity measures are undecidable, thus demonstrating a limitation of the approach towards learnability and model selection via such complexity measures.

\subsection{Related Work}\label{SbSct:related-work}
\cite{Lathrop.1996} made an early investigation into the relationship between computability and learnability. The main question in \cite{Lathrop.1996} is whether and under which notions of ``learnability'' one can consider an uncomputable problem to be learnable. More precisely, \cite{Lathrop.1996} considered the task of learning the halting problem relative to an oracle.

Both \cite{schaefer1999deciding} and \cite{Zhao.2018} studied the computability of finiteness of the VC-dimension. In particular, Theorem $1$ in \cite{Zhao.2018} and Theorem $4.1$ in \cite{schaefer1999deciding} state: Deciding finiteness of the VC-dimension of a computable concept class is $\Sigma_2$-complete. This implies our Corollary \ref{CrlTuringUndecidabilityFinitenessVCDim}, the Turing undecidability of finiteness of the VC-dimension. The proofs of \cite{schaefer1999deciding} and \cite{Zhao.2018} used that deciding finiteness of the domain of a computable function is $\Sigma_2$-complete (see, e.g., Theorem IV.$3.2$ in \cite{soare1978recursively}). \cite{Zhao.2018} additionally invoked a result by \cite{Laskowski.1992}. 
While one of our results is already implied by \cite{schaefer1999deciding} and \cite{Zhao.2018}, we consider our work to be a significant extension in two directions: On the one hand, we consider both Turing and Gödel undecidability. On the other hand, our proof strategies are at the same time more direct, using no results from logic beyond Gödel's incompleteness theorems and the Turing undecidability of the halting problem, and are flexibly applicable to other complexity measures and learning scenarios.
Note that, while \cite{schaefer1999deciding, Zhao.2018} as well as the present paper consider function classes with discrete instance spaces, \cite{calvert2015pac} established $\Sigma_3$-completeness of finiteness of the VC-dimension among a certain family of reasonable classes that in particular can have continuous input spaces.

\cite{BenDavid.2019} proposed the ``estimating-the-maximum'' (EMX) problem and proved that learnability in this model is independent of the ZFC axioms. While this already indicates that learning can be undecidable, our results add new insight in at least two ways. First, our results are for already established learning models. In particular, whereas \cite{BenDavid.2019} showed that, assuming consistency of ZFC, there is no dimension-like quantity of finite character that characterizes EMX learnability, our results include scenarios in which such dimensions for learning exist. Second, whereas \cite{BenDavid.2019} used the continuum, the continuum hypothesis, and the axiom of choice, we only use natural numbers and computable objects. This allows us to prove uncomputability results, which cannot be derived from the results of \cite{BenDavid.2019}. Some implications and some limitations of the approach of \cite{BenDavid.2019} have been discussed, e.g., in \cite{Hart.2019,Taylor.2019,Gandolfi.2020}. Building on \cite{BenDavid.2019}, \cite{hanneke2023bandit} recently established a similar-in-spirit undecidability result for bandit learnability.

\cite{Agarwal.2020} initiated a study of computable learners, which then truly deserve to be called ``learning algorithms.'' In particular, \cite{Agarwal.2020} showed that not every PAC learnable class admits a computable learner, identified conditions under which PAC learnability implies computable PAC learnability, and, in~\cite{agarwal2021openproblem}, posed several open questions regarding computable learning, some of which have been resolved in \cite{sterkenburg2022characterizations}. Thereby, \cite{Agarwal.2020} extended considerations from \cite{Soloveichik.2008}, which studied the task of non-uniform learning over all computable functions by a computable learner. As the underlying questions of \cite{Agarwal.2020} and our work differ, the results are not comparable. However, as we show the function classes $\mathcal{G}_F$ and $\mathcal{H}_M$ to be computable, the results of \cite{Agarwal.2020} imply that our undecidability results hold not only for PAC learning, but also for computable PAC learning. The notion of computable learners was extended by~\cite{ackerman2021computable} to a broader class of input and output spaces, so-called computable extended metric spaces. In related work,~\cite{crook2021computability} considered computably verifiable properties of learners in classification tasks. In the spirit of these developments, \cite{hasrati2023computable} recently introduced a computable variant of online learning.

\cite{sehra2021undecidability} took yet another perspective on undecidabity in learning theory. Namely, \cite{sehra2021undecidability} considered the problem of deciding, given an algorithm $\mathcal{A}$ and a dataset $d$, whether $\mathcal{A}$ is a learning algorithm and the output model of $\mathcal{A}$ underfits $d$. Here, \cite{sehra2021undecidability} used an information-theoretic notion of underfitting. The main result of \cite{sehra2021undecidability}: This decision problem can be reduced to the halting problem and is thus Turing undecidable.

\cite{hanneke2021universal} recently identified a further potential source of undecidability in learning theory. They studied the existence of universally Bayes consistent learners for countable multiclass classification, i.e., of learners whose classification error almost surely converges to the optimal Bayes risk (over all Borel measurable classifiers) as the sample size goes to infinity. Theorem $4.1$ in \cite{hanneke2021universal} states: A universally Bayes consistent classifier can only exist if the metric space from which the instances are drawn is essentially separable. The existence of metric spaces that are not essentially separable, however, is believed to be independent of the ZFC axioms. Hence, whether all metric instance spaces admit a universally Bayes consistent classifier might turn out to be a learning-theoretic question independent of ZFC.

Combinatorial complexity measures for learning have also been studied in computational complexity theory. \cite{papadimitriou1996limited}, motivated by \cite{linial1991results}, determined the complexity of computing the VC-dimension of a finite concept class over a finite domain. While they argue that this problem is probably not NP-complete, \cite{papadimitriou1996limited} proved its completeness for the complexity class \LOGNP, a logarithmically-restricted version of NP. \cite{shinohara1993complexity} obtained a similar completeness result. \cite{frances1998optimal} then, by reduction to computing the VC-dimension, established the \LOGNP-hardness of computing the Littlestone dimension. \cite{schaefer1999deciding} later showed that a variant of the above problem, namely that of computing the VC-dimension of a class described by a polynomial-sized circuit, is $\Sigma_3^p$-complete. Computing the Littlestone dimension from a circuit description is \PSPACE-complete~\cite{schaefer2000deciding}. Extending a result by \cite{schaefer1999deciding}, \cite{mossel2002complexity} determined the complexity of a promise version of approximating the VC-dimension of a class associated to a polynomial-size circuit. More recently, \cite{manurangsi2017inapproximability} has proved nearly tight quasi-polynomial time lower bounds for approximating the VC-dimension and the Littlestone dimension, assuming the randomized Exponential Time Hypothesis.

\subsection{Structure of the Paper}
We use Section \ref{SctBasicsLearningTheory} to introduce our different learning scenarios. 
Section \ref{SctUndecidableNontrivialityFiniteness} contains our main constructions, which we use to reprove the folklore results of undecidability of non-triviality and finitenes, here for function classes. 
In Section \ref{SctUndecidableLearnability}, we demonstrate that our constructions directly yield undecidability of learnability in different learning scenarios. \ref{SbSctUndecidablePACClassification} focuses on PAC binary classification, Subsection \ref{SbSctUndecidableTeaching} considers teacher-learner interactions, and Subsection \ref{SbSctUndecidableOnlineLearning} discusses both uniform and universal online learning.
We conclude with an outlook and open questions in Section \ref{SctConclusion}. Full proofs appear either directly in the text or in Appendix \ref{AppendixProofs}. Appendices \ref{SctGoedelPreliminaries} and \ref{SctTuringPreliminaries} contain standard definitions and results related to formal systems and computability that are used in the main text.

\section{Preliminaries on Learning Theory}\label{SctBasicsLearningTheory}

\subsection{PAC Binary Classification and the VC-Dimension}\label{SbSctPreliminariesVCDim}

We start by recalling one of the most influential learning models for binary classification:

\begin{definition}[Probably approximately correct binary classification \cite{Valiant.1984}]\label{DffPAC}
Let $\mathcal{X}$ be some space, write $\mathcal{Z}=\mathcal{X}\times\{ 0,1\}$. Let $\mathcal{G}\subset\{0,1\}^\mathcal{X}$, and let $D$ be a probability distribution on $\mathcal{Z}$. A map $\mathcal{A}:\bigcup_{m=1}^\infty\mathcal{Z}^m\to\{ 0,1\}^\mathcal{X}$, $S\mapsto h_S$, is a \emph{probably approximately correct (PAC)} learner for $\mathcal{G}$ if there exists a function $m:(0,1)^2\to \IN_{\geq 1}=\{1,2,3,\ldots\}$ such that, given $\varepsilon,\delta\in (0,1)$, if $m\geq m(\varepsilon,\delta)$, then, with probability $\geq 1-\delta$ with respect to repeated sampling of $S\sim D^m$, it holds that $$\mathbb{P}_{(x,y)\sim D} [h_S(x)\neq y]\leq\varepsilon + \inf\limits_{g\in\mathcal{G}} \mathbb{P}_{(x,y)\sim D} [g(x)\neq y].$$
\end{definition}

The PAC learners of interest are \emph{polynomial PAC learners}, for which the sample size $m(\varepsilon,\delta)$ can be chosen to depend polynomially on $\nicefrac{1}{\varepsilon}$ and $\log\left(\nicefrac{1}{\delta}\right)$. Here, the ``polynomial'' refers to the sample size only, not to the runtime. If $\mathcal{G}$ admits a polynomial PAC learner, we call $\mathcal{G}$ \emph{PAC learnable}.

For the scenario of binary classification, whether there exists a polynomial PAC learner can be understood in terms of a combinatorial quantity of the function class under consideration.

\begin{definition}[VC-dimension \cite{Vapnik.1971}]\label{DffVCDim}
Let $\mathcal{G}\subseteq\lbrace 0,1\rbrace^\mathcal{X}$. The Vapnik-Chervonenkis dimension, abbreviated as VC-dimension, of $\mathcal{G}$ is defined to be 
\begin{align*}
\VC(\mathcal{G}):=\sup\lbrace n\in\IN~|~\exists S\subseteq\mathcal{X}:\lvert S\rvert=n~\wedge ~ \left\lvert\mathcal{G}\rvert_{S}\right\rvert=2^n\rbrace,
\end{align*}
with $\mathcal{G}\rvert_{S}$ the restriction of $\mathcal{G}$ to $S$.
If $S\subseteq\mathcal{X}$ is s.t.~$\left\lvert\mathcal{G}\rvert_{S}\right\rvert=2^{\lvert S\rvert}$, we say that $S$ is shattered by $\mathcal{G}$.
\end{definition} 

Under suitable measurability assumptions on the function class $\mathcal{G}$, we have the following

\begin{theorem}[Fundamental theorem of binary classification (see, e.g., \cite{ShalevShwartz.2019})]\label{ThmFundamentalTheoremBinaryClassification}
Let $\mathcal{G}\subset\{0,1\}^\mathcal{X}$. 
$\mathcal{G}$ is PAC learnable if and only if $\operatorname{VCdim}(\mathcal{G})<\infty$.
\end{theorem}

Among the assumptions on $\mathcal{G}$ that guarantee the equivalence in Theorem \ref{ThmFundamentalTheoremBinaryClassification} are that $\mathcal{G}$ be image admissible Suslin, universally separable, well behaved, or countable \cite{dudley1978central, Pollard.1984, blumer1989learnability, pestov2011pac}.
The function classes considered in this paper are all countable, so Theorem \ref{ThmFundamentalTheoremBinaryClassification} applies. 
Therefore, when studying PAC learnability, we focus on studying finiteness of the VC-dimension. Interestingly, \cite{arunachalam2018optimal} found the VC-dimension to characterize quantum PAC learnability in the same way. Therefore, our undecidability results carry over to quantum PAC learnability.

\subsection{Teaching Problems and the Teaching Dimension}\label{SbSctPreliminariesTDim}
We now turn our attention to a different learning problem. The differences to the PAC model are two-fold. The source of the training data is now a benevolent teacher who knows the function to be learned. And, instead of requiring the learner to approximate the unknown function with high probability, the unknown function must be exactly identified. To help the learner identify the target function, the teacher has to provide a training data set that uniquely characterizes it. The difficulty of the learning/teaching problem is then captured by the worst case size of a smallest such training data set. This is made formal in the following

\begin{definition}[Teaching sets and the teaching dimension \cite{Goldman.1995}]\label{DffTeachingDim}
Let $\mathcal{G}\subset\{0,1\}^\mathcal{X}$, $g\in \mathcal{G}$. A set $S=\{(x_i,y_i)\}_{i=1}^N\subset\mathcal{X}\times\{0,1\}$, $N\in\IN\cup\{\infty\}$, is a \emph{teaching set} for $g$ in $\mathcal{G}$ if $g(x_i)=y_i$ for all $(x_i,y_i)\in S$ and for every $\tilde{g}\in\mathcal{G}\setminus\{g\}$ there exists $(x_j,y_j)\in S$ such that $\tilde{g}(x_j)\neq y_j$. That is, $g$ is the unique concept in $\mathcal{G}$ that is consistent with the labelled data $S$.

The \emph{teaching dimension} of $\mathcal{G}$ is the worst case size of a minimal teaching set, i.e.,
\begin{align*}
\operatorname{Tdim}(\mathcal{G}):=\sup\limits_{g\in\mathcal{G}}\inf\{\lvert S\rvert~|~S \textrm{ is a teaching set for }g\}.
\end{align*}
\end{definition}

We consider a learning/teaching problem for a class $\mathcal{G}$ to be solvable if $\operatorname{Tdim}(\mathcal{G})<\infty$. Note that we will use this notion specifically for $\mathcal{X}=\IN$. This is non-standard. Usually, $\mathcal{X}$ is assumed to be finite so that the teaching dimension is automatically finite.

If the teacher and the learner are allowed to make additional assumptions about the respectively other party's strategy, more refined notions of teaching dimensions should be used (see \cite{Zilles.2011} for an overview). We, however, restrict our attention to the simplest complexity measure for teaching tasks, namely the one in Definition \ref{DffTeachingDim}.

\subsection{Online Learning and Littlestone Trees and Dimension}\label{SbSctPreliminariesLDim}

In online learning, we consider a game between two players, a learner $L$ and an adversary $A$, both of which know the function class $\mathcal{G}\subseteq\{0,1\}^\mathcal{X}$. The game consists of infinitely many rounds. Round $t\in\mathbb{N}_{\geq_1}$ consists of three steps: First, $A$ chooses a ``question'' $x_t\in\mathcal{X}$. Second, $L$ guesses a label $\hat{y}_t\in\{0,1\}$. Third, $A$ reveals the true label $y_t\in\{0,1\}$ to $L$. Crucially, $A$ must ensure that the sequence of true labels can actually be realized within $\mathcal{G}$. That is, the produced sequence $\left((x_t,y_t)\right)_{t=1}^\infty$ must be such that, for every $t\in\mathbb{N}_{\geq 1}$, there exists a function $g\in\mathcal{G}$ with $g(x_s)=y_s$ for all $1\leq s\leq t$. Note: $A$ does not have to pick a fixed $g\in\mathcal{G}$ in advance. Instead $A$ can choose the true labels adaptively, based on the actions of $L$ and $A$ in previous rounds.

The goal of $L$ is to make as few mistakes as possible, where we say that $L$ makes a mistake in round $t\in\mathbb{N}_{\geq 1}$ if $\hat{y}_t\neq y_t$. Conversely, $A$ wants to make the number of mistakes made by $L$ as large as possible. Note that, while we can also interpret teaching problems as two-player games, the role of the second player is quite different. A teacher is seen as benevolent and has the same goal as the learner. In contrast, an adversary's goal is exactly opposite to that of the learner.

We consider two variants of the online learning problem. On the one hand, we work in the scenario of \emph{universal online learning}, recently introduced in \cite{Bousquet.2020}. We say that $\mathcal{G}\subseteq\{0,1\}^\mathcal{X}$ is \emph{universally online learnable} if there exists an adaptive strategy $\hat{y}_t = \hat{y}_t(x_1,y_1,\ldots,x_{t-1},y_{t-1},x_t)$ for $L$ such that, for any adversary $A$, $L$ makes only finitely many mistakes in the above game. On the other hand, we also formulate results in the uniform mistake bound model of online learning, which we refer to as \emph{uniform online learning}. We say that $\mathcal{G}\subseteq\{0,1\}^\mathcal{X}$ is \emph{uniformly online learnable} if there exist a $d\in\IN$ and an (adaptive) strategy $\hat{y}_t = \hat{y}_t(x_1,y_1,\ldots,x_{t-1},y_{t-1},x_t)$ for $L$ such that, for any adversary $A$, $L$ makes at most $d$ mistakes in the above game.

Both whether a class is universally or uniformly online learnable can be understood in terms of so-called \emph{Littlestone trees}.

\begin{definition}[Littlestone trees \cite{Littlestone.1988, Bousquet.2020}]\label{DffLittlestoneTree}
A set of points $\{x_{\mathbf{v}}\}_{\mathbf{v}\in\{0,1\}^k,0\leq k <d}\subset\mathcal{X}$ is a \emph{Littlestone tree of depth} $d\leq\infty$ of a binary-valued function class $\mathcal{G}\subseteq\{0,1\}^\mathcal{X}$ if, for every $(y_i)_{i=0}^d$ and for every $0\leq n<d$, there exists $g\in\mathcal{G}$ such that $g(x_{y_0\ldots y_k})=y_{k+1}$ holds for all $0\leq k\leq n$. 
We say that $\mathcal{G}$ \emph{has an infinite Littlestone tree} if there exists a Littlestone tree of depth $\infty$ of $\mathcal{G}$.
\end{definition}

A Littlestone tree of $\mathcal{G}$ is a complete binary tree in which the nodes are labelled by points in $\mathcal{X}$ and the edges are labelled by $0$ or $1$ in such a way that for every path of finite length, starting from the root of the tree, there is a function in $\mathcal{G}$ that labels all nodes along the path according to the respectively outgoing edges. Note that the definition is only concerned with finite paths, even for an infinite Littlestone tree. 

The relation between universal online learnability and the (non-)existence of infinite Littlestone trees is summarized in the following

\begin{theorem}[{\cite[Theorem $3.1$]{Bousquet.2020}}]\label{ThmUniversalOnlineLearningNoInfiniteLittlestoneTrees}
$\mathcal{G}\subseteq\{0,1\}^\mathcal{X}$ is universally online learnable iff $\mathcal{G}$ does not have an infinite Littlestone tree.
\end{theorem}

Going from ``universal'' to ``uniform'' on the level of Littlestone trees corresponds to requiring a uniform bound on the depth of all Littlestone trees of a class. This gives rise to

\begin{definition}[Littlestone dimension \cite{Littlestone.1988}]
Let $\mathcal{G}\subseteq\lbrace 0,1\rbrace^\mathcal{X}$. The Littlestone dimension of $\mathcal{G}$ is defined to be 
\begin{align*}
    \operatorname{Ldim}(\mathcal{G}) 
    := \sup\left\{d\in\mathbb{N}_0~|~\mathcal{G}\textrm{ has a Littlestone tree of depth } d \right\}.
\end{align*}
\end{definition}

Note that, if $\mathcal{G}$ has an infinite Littlestone tree, then $\operatorname{Ldim}(\mathcal{G})=\infty$. The converse, however, is not true, as $\operatorname{Ldim}(\mathcal{G})=\infty$ also holds if $\mathcal{G}$ has Littlestone trees of arbitrarily large depth but no infinite Littlestone tree.

The Littlestone dimension characterizes uniform online learnability according to the following

\begin{theorem}[{\cite[Theorem $3$]{Littlestone.1988}}] \label{ThmUniformOnlineLearningFiniteLDim}
$\mathcal{G}\subseteq\{0,1\}^\mathcal{X}$ is uniformly online learnable with at most $d\in\mathbb{N}$ mistakes iff $\operatorname{Ldim}(\mathcal{G})\leq d$. In particular, $\mathcal{G}\subseteq\{0,1\}^\mathcal{X}$ is uniformly online learnable iff $\operatorname{Ldim}(\mathcal{G})<\infty$.
\end{theorem}

\section{Undecidability of Non-Triviality and Finiteness}\label{SctUndecidableNontrivialityFiniteness}

\subsection{Gödel Undecidability}\label{SbSctGoedelUndecidableNontrivialityFiniteness}

For the purpose of this subsection, let $F$ denote a recursively enumerable formal system in which infinitely many different theorems can be proved. (See Definition \ref{DffRecursivelyEnumerableFormalSystems} for a definition of ``recursively enumerable.'') Let $\varphi$ be a primitive recursive enumeration of the theorems provable in $F$. Here, we think of theorems being ``different'' in a symbolic way. That is, two theorems are the same if and only if they are the exact same sequence of symbols from the alphabet available in $F$. This, in turn, is equivalent to the two theorems having the same Gödel number in a fixed Gödel numbering. 

Also, we will denote by $E^2:\IN\to\IN\times\IN$ a primitive recursive enumeration of $\IN^2$. That is, $E^2$ is a total bijective function such that both component functions $E^2_i:\IN\to\IN$, $i=1,2$, are primitive recursive and such that the inverse $(E^2)^{-1}:\IN\times\IN\to\IN$ is primitive recursive. The existence of such an $E^2$ can, e.g., be proved using so-called pairing functions.

We begin by defining our main object of study for this subsection.

\begin{definition}\label{DffGoedelFctClass}
Let $F$, $\varphi$, $E^2$ be as above. We denote the space of finitely supported sequences with elements in $\{0,1\}$ by $c_c (\{ 0,1\}):=\{(b(k))_{k\in\IN}\in\{0,1\}^{\IN}~|~b(k)\neq 0\textrm{ only finitely often}\}$. For such a finitely supported $a=(a(k))_{k\in\IN}\in c_c (\{ 0,1\})$, define the function $g_a : \IN\to\{ 0,1\}$ via
\begin{align*}
g_a(n)=\begin{cases} a(n) \quad&\textrm{if }\varphi(E^2_1 (n))=\neg \varphi(E^2_2 (n)) \\ 0 &\textrm{else}\end{cases},
\end{align*}
and the function class $\mathcal{G}_F:= \{g_a\}_{a\in c_c(\{0,1\})}.$
\end{definition}

In this formulation, we think of $a$ as a parameter, determining an element of $\mathcal{G}_F$, and we define the class $\mathcal{G}_F$ in terms of a prescription for evaluating the function with parameter $a$ on an input $x$.
Here, the equality $\varphi(E^2_1 (n))=\neg \varphi(E^2_2 (n))$ is to be understood as the symbolic equality between the theorem with Gödel number $\varphi(E^2_1 (n))$ and the negation of the theorem with Gödel number $\varphi(E^2_2 (n))$. Equivalently, we require equality of the corresponding Gödel numbers.
In words, $\mathcal{G}_F$ is the class of all $\{0,1\}$-valued finitely supported sequences that have non-zero entries only on natural numbers indexing a pair of inconsistent theorems in $F$.

We first observe that the class $\mathcal{G}_F$ ``collapses'' to a single function, the zero function, if and only if the underlying formal system $F$ is consistent.

\begin{proposition}\label{PrpConsistencyCollapseFctClass}
$F$ is consistent iff $\mathcal{G}_F=\{ 0\}$.
\end{proposition}
\begin{proof}
This follows from the construction of the function class because $E^2$ is surjective and the range of $\varphi$ consists exactly of all Gödel numbers of theorems provable in $F$.
\end{proof}

In contrast, if the underlying formal system $F$ is inconsistent, then $\mathcal{G}_F$ is highly complex:

\begin{theorem}\label{ThmInconsistencyHighComplexity}
    If $F$ is inconsistent, there exists a sequence $(n_i)_{i=1}^\infty$ such that $\{n_1,\ldots,n_N\}$ is shattered by $\mathcal{G}_F$ for all $N\in\mathbb{N}$.
\end{theorem}

For the proof, we first recall that ``anything can be deduced from a contradiction,'' also known as ``ex falso quodlibet.''
\begin{proposition}\label{PrpExFalsoQuodlibet}
    Let $F$ be an inconsistent formal system. Let $q$ be a theorem in $F$. Then both $q$ and $\neg q$ can be proved in $F$.
    \end{proposition}
\begin{proof}
    See Appendix \ref{AppendixProofs}.
\end{proof}

With this we can now prove Theorem \ref{ThmInconsistencyHighComplexity}.

\begin{proof}[Proof of Theorem \ref{ThmInconsistencyHighComplexity}]
    As $F$ is inconsistent and infinitely many different theorems can be proved in $F$, by ``ex falso quodlibet'' there are infinitely many $n\in\IN$ such that $\varphi (E^2_1 (n))=\neg\varphi (E^2_2 (n))$, because $E_1^2$ and $E_2^2$ are surjective and the range of $\varphi$ consists exactly of all Gödel numbers of theorems provable in $F$.
    
    Let $N\in\IN$. Then, by the above, there exist pairwise distinct $n_1,\ldots,n_N\in\IN$ such that $\varphi (E^2_1 (n_i))=\neg\varphi (E^2_2 (n_i))$ for all $1\leq i\leq N$. Let $b\in\{ 0,1\}^N$ be arbitrary. Define $a_b\in c_c(\{0,1\})$ as
    \begin{align*}
        a_b (n_i) = b(i)~\textrm{for }1\leq i\leq N ,\quad a_b(n) = 0~\textrm{for } n\in\IN\setminus\{n_1,\ldots,n_N\}.
    \end{align*}
    Then we clearly have $g_{a_b} (n_i) = b(i)$ for all $1\leq i\leq N$. Therefore, we have shown $\mathcal{G}_F\rvert_{\{n_1,\ldots,n_N\}} = \{0,1\}^{\{n_1,\ldots,n_N\}}$. In other words, $\{n_1,\ldots,n_N\}$ is shattered by $\mathcal{G}_F$.
\end{proof}

\begin{remark}\label{RmkTrivialGoedelClass}
    There is a na\"ive way of constructing a function class that satisfies the same properties as the ones just established for $\mathcal{G}_F$. 
    Namely, given $F$, we could define
    \begin{align*}
        \tilde{\mathcal{G}}_F:=
        \begin{cases}
        \{ 0\}\quad &\textrm{ if } F\textrm{ is consistent}\\
        c_c(\{0,1\}) &\textrm{ else }
        \end{cases}.
    \end{align*}
    This definition of $\tilde{\mathcal{G}}_F$, however, lacks the constructive aspect of Definition \ref{DffGoedelFctClass}, which comes with a concrete instruction for how to evaluate any element of the function class on any given input.
    Thus, whereas we can understand $F\mapsto\mathcal{G}_F$ as a computable mapping (see Corollary \ref{CrlGoedelMappingComputable}), the same is not the case for $F\mapsto\tilde{\mathcal{G}}_F$. 
\end{remark}

If the formal system $F$ is capable of expressing the class $\mathcal{G}_F$ and its triviality, we can combine Proposition~\ref{PrpConsistencyCollapseFctClass} and Theorem~\ref{ThmInconsistencyHighComplexity} with Gödel's second incompleteness theorem to obtain:

\begin{corollary}\label{CrlGoedelUndecidabilityNonTriviality}
    Assume that $F$ is a recursively enumerable and consistent formal system that contains elementary arithmetic, that can reason about $\mathcal{G}_F$, and such that infinitely many different theorems can be proved in $F$. 
    Then $\mathcal{G}_F = \{0\}$, but this cannot be proved in $F$.
\end{corollary}
\begin{proof}
    Assume for contradiction that the statement $\mathcal{G}_F = \{0\}$ can be proved in $F$. With Proposition~\ref{PrpConsistencyCollapseFctClass} and Theorem~\ref{ThmInconsistencyHighComplexity}, we have given a proof that this implies consistency of $F$. If this proof can be expressed in the formal system $F$, $F$ proves its own consistency. This contradicts Gödel's second incompleteness theorem.
\end{proof}

Corollary \ref{CrlGoedelUndecidabilityNonTriviality} shows that the constructive description of $\mathcal{G}_F$ via a rule for evaluating any of its elements on any possible input makes it impossible to detect when $\mathcal{G}_F$ becomes trivial. This constitutes a limitation of such a kind of access to a function class.

Now we come to the second relevant observation about the class $\mathcal{G}_F$: Not only is it a computable function class, but even the mapping $F\mapsto\mathcal{G}_F$ is computable. We first prove the slightly weaker result that $\mathcal{G}_F$ is a computable function class in the sense of Definition \ref{DffComputableFunctionClass}:

\begin{theorem}\label{ThmGoedelBinaryFctClassComputable}
    Assume that $F$ is a recursively enumerable formal system. Then the class $\mathcal{G}_F$ is computable.
\end{theorem}

As a first step towards proving this result, we observe that the sequence space $c_c(\{0,1\})$ used for indexing the class can be recursively enumerated.

\begin{lemma}\label{LmmRecursiveEnumerabilityCompactlySupportedBinarySequences}
    There exists a primitive recursive function $C:\IN\times\IN\to\{ 0,1\}$ that enumerates $c_c(\{0,1\})$, i.e., such that $c_c(\{ 0,1\})=\{ n\mapsto C(m,n)~|~m\in\IN\}$. 
\end{lemma}
\begin{proof}
    See Appendix \ref{AppendixProofs}.
\end{proof}

With this ingredient at hand, we can prove Theorem \ref{ThmGoedelBinaryFctClassComputable}.

\begin{proof}[Proof of Theorem \ref{ThmGoedelBinaryFctClassComputable}]
    According to Definition \ref{DffComputableFunctionClass}, we want to find a total computable function $G_F:\IN\times\IN\to\{ 0,1\}$ such that $\mathcal{G}_F=\{n\mapsto G_F(m,n)~|~m\in\IN\}$. We define
    \begin{align*}
    G_F(m,n):=
    \begin{cases}
    C(m,n) \quad &\textrm{ if } \varphi(E^2_1 (n))=\neg \varphi(E^2_2 (n))\\
    0 &\textrm{ else} 
    \end{cases}.
    \end{align*}
    Since $C$ recursively enumerates $c_c(\{0,1\})$, we indeed have $\mathcal{G}_F=\{n\mapsto G_F(m,n)~|~m\in\IN\}$. It remains to show that $G_F$ is a total computable function. As $C$ is total computable, even primitive recursive by Lemma \ref{LmmRecursiveEnumerabilityCompactlySupportedBinarySequences}, it suffices to show that the predicate $\varphi(E^2_1 (n))=\neg \varphi(E^2_2 (n))$ is total computable.
    
    To this end, recall that $E^2_1$, $E^2_2$ and $\varphi$ are primitive recursive. Thus, we only have to show that, given the Gödel numbers of two theorems, checking whether the theorem corresponding to the first number is the negation of the theorem corresponding to the second number can be done in a computable manner. This is even possible in a primitive recursive manner simply by how Gödel numbers are constructed.
\end{proof}

Note that our proof of Theorem \ref{ThmGoedelBinaryFctClassComputable} even shows that $\mathcal{G}_F$ is primitive recursive if we define a primitive recursive class of functions analogously to Definition \ref{DffComputableFunctionClass}.
The proof tells us more about the construction of $\mathcal{G}_F$ with respect to computability. Not only is the function class $\mathcal{G}_F$ computable for every formal system $F$. (This is also true for $\tilde{\mathcal{G}}_F$.) But we even see that the assignment $F\mapsto \mathcal{G}_F$ is computable in the following sense:

\begin{corollary}\label{CrlGoedelMappingComputable}
    There exists a partial computable function $\mathbb{G}:\IN^3\to\IN$ such that $\mathcal{G}_F=\{\IN\ni n\mapsto \mathbb{G}(\varphi,m,n)~|~m\in\IN\}$ for any recursively enumerable formal system $F$ whose theorems are enumerated by the primitive recursive function $\varphi:\IN\to\IN$.
\end{corollary}
\begin{proof}[Proof sketch]
    As $\varphi$ is primitive recursive, it is in particular computable. Thus, we can represent it via its code with respect to our universal Turing machine. With this code, we can compute the predicate $\varphi(E^2_1 (n))=\neg \varphi(E^2_2 (n))$ and the Corollary is proved just like Theorem \ref{ThmGoedelBinaryFctClassComputable}.
\end{proof}

Together, Theorem \ref{ThmGoedelBinaryFctClassComputable} and Corollary \ref{CrlGoedelMappingComputable} provide an advantage of our construction over the ``trivial'' $\tilde{G}_F$ in Remark \ref{RmkTrivialGoedelClass}. Given a recursively enumerable system in terms of an explicit primitive recursive enumeration $\varphi$ of theorems, they provide us with an explicit algorithmic procedure for evaluating elements of the function class $\mathcal{G}_F$ and thereby with an explicit description of $\mathcal{G}_F$ obtained by fixing certain inputs of the concrete function $\mathbb{G}$. In this way, Definition \ref{DffGoedelFctClass} can be viewed as defining $\mathcal{G}_F$ in terms of a function that upon input of a ``parameter'' $a$ and an input $x$ outputs the value $g_a (x)$, and whose action is computable from the underlying formal system $F$. 

\subsection{Turing Undecidability}\label{SbSctTuringUndecidableNontrivialityFiniteness}

We now change the perspective and ask whether there is a general-purpose algorithmic procedure for deciding whether a binary-valued function class is simple, here understood as finite, or has high complexity. We begin by describing what such a hypothetical algorithm should do: It would take as input the code of a computable binary-valued function class $\mathcal{G}$, coming from a uniformly computable family of such classes. It should output $0$ if $\mathcal{G}$ shatters an infinite set and $1$ if $\mathcal{G}$ is finite. Note that such an algorithm would work ``only'' for computable function classes since it is exactly the computability which allows us to provide their code as input. As such, the notion of a computable function class introduced in Definition \ref{DffComputableFunctionClass} is crucial, since it allows us to identify such a class with a computable function, whose code can then be fed as input to an algorithm.

We show that such an algorithm does not exist by reduction to the halting problem. The ``encoding'' of the halting problem into the finiteness versus high complexity of a function class is achieved by the following construction.

\begin{definition}\label{DffTuringFctClass}
    Let $M$ be a finite-state Turing machine. For a finitely supported sequence $a=(a(k))_{k\in\IN}\in c_c(\{0,1\})$, define the function $h_a : \IN\to\{ 0,1\}$ via
    \begin{align*}
        h_a(n)=\begin{cases} a(n) \quad&\textrm{if } M\textrm{ does not halt after } \leq n\textrm{ steps on the empty input} \\ 0 &\textrm{else}\end{cases},
    \end{align*}
    and the function class $\mathcal{H}_{M}:= \{h_a\}_{a\in c_c(\{ 0,1\})}.$
\end{definition}

Similarly to Definition \ref{DffGoedelFctClass}, also Definition \ref{DffTuringFctClass} defines a function class in terms of a prescription for evaluating any specific element of the class on any given input. Below, we will see that this prescription is computable from the underlying Turing machine $M$.

We first provide an alternative expression for the class $\mathcal{H}_{M}$:

\begin{proposition}\label{PrpTuringClassAlternativeForm}
    The class $\mathcal{H}_{M}$ from Definition~\ref{DffTuringFctClass} is given as
    \begin{align*}
        \mathcal{H}_M:=
        \begin{cases} \{0,1\}^{\{0,\ldots,K-1\}}\quad &\textrm{ if } M\textrm{ halts after exactly } K\textrm{ steps on the empty input}\\ c_c(\{0,1\}) &\textrm{ else} \end{cases},
    \end{align*}
    where we think of $\{0,1\}^{\{0,\ldots,K-1\}}$ as being embedded into $\{0,1\}^\IN$ as the first $K$ sequence elements, to which we append zeros.
\end{proposition}

For a reduction to the halting problem, we need to establish two claims. 
First, we need to show that $\mathcal{H}_M$ is computable according to Definition \ref{DffComputableFunctionClass}, so that it makes sense to talk about $\mathcal{H}_M$ as input to a hypothetical algorithm that decides finiteness. Only then will $\mathcal{H}_M$, or more precisely the corresponding function $H_M$, have a code that we can use as input for our hypothetical decision algorithm. 
Second, we need to show that constructing the class $\mathcal{H}_M$ from the Turing machine $M$ can be done in a computable way. That is, we need to prove that the mapping $M\mapsto\mathcal{H}_M$ is computable. We begin by establishing computability of $\mathcal{H}_M$.

\begin{theorem}\label{ThmTuringFctClassComputable}
Let $M$ be a Turing machine. The function class $\mathcal{H}_M$ is computable.
\end{theorem}
\begin{proof}
We have already seen in Lemma \ref{LmmRecursiveEnumerabilityCompactlySupportedBinarySequences} that there exists a primitive recursive function $C:\IN\times\IN\to\{ 0,1\}$such that $c_c(\{ 0,1\})=\{ m\mapsto C(m,n)~|~m\in\IN\}$. Therefore, if we define 
\begin{align*}
H_M(m,n) = \begin{cases} C(m,n) \quad &\textrm{ if } M\textrm{ does not halt after } \leq n\textrm{ steps on the empty input}\\0 &\textrm{ else } \end{cases},
\end{align*}
then $\mathcal{H}_M=\{n\mapsto H_M(m,n)~|~m\in\IN\}$. Moreover, $H_M$ is a computable function because is defined from computable functions and a case distinction with a computable predicate. Hence, $\mathcal{H}_M$ is a computable function class according to Definition \ref{DffComputableFunctionClass}.
\end{proof}

Computability of $\mathcal{H}_M$ can be seen more easily: $\mathcal{H}_M$ is either finite and thus trivially computable or it is equal to $c_c(\{0,1\})$ and thus computable by Lemma \ref{LmmRecursiveEnumerabilityCompactlySupportedBinarySequences}. We present the proof above because, similarly to our reasoning in Subsection \ref{SbSctGoedelUndecidableNontrivialityFiniteness}, it already gives us the computability of $M\mapsto\mathcal{H}_M$:

\begin{corollary}\label{CrlTuringMappingComputable}
There exists a partial computable function $\mathbb{H}:\IN^3\to\IN$ such that $\mathcal{H}_M=\{\IN\ni n\mapsto \mathbb{H}(M,m,n)~|~m\in\IN \}$ for any Turing machine $M$.
\end{corollary}

The computability of $M\mapsto\mathcal{H}_M$ is crucial for the final step in our proof of Turing undecidability. And it provides an explicit description of the class $\mathcal{H}_M$ obtained by fixing ``input parameters'' of the function $\mathbb{H}$.
This relatively simple proof of computability for the mapping $M\mapsto\mathcal{H}_{M}$ is also the reason why we introduced $\mathcal{H}_M$ as in Definition~\ref{DffTuringFctClass} rather than with the simpler expression from Proposition~\ref{PrpTuringClassAlternativeForm}. 

Now we have everything we need to finish the reduction to the halting problem and thereby our proof of Turing undecidability.

\begin{corollary}\label{CrlTuringUndecidabilityFiniteness}
    There is no Turing machine that, upon input of the code of an arbitrary computable binary-valued function class coming from $\{\mathcal{H}_M\}_{M}$, decides whether that class is finite or equal to $c_c(\{0,1\})$.
\end{corollary}
\begin{proof}
    Assume for contradiction that there is such a Turing machine $M_{\rm Fin}$. Then we could construct a Turing machine for solving the halting problem on the empty input as follows:

    Given as input the code of a Turing machine $M$, compute the code of the corresponding class $\mathcal{H}_{M}$, or, more precisely, the function $H_M$. This step is possible because the code of a concatenation of Turing machines is a primitive recursive function of their respective codes and because the mapping $M\mapsto\mathcal{H}_M$ is computable by Corollary \ref{CrlTuringMappingComputable}. Now feed that code to the Turing machine $M_{\rm Fin}$. If it outputs $1$ output, ``yes, halts,'' otherwise output ``no, doesn't halt.''

    As the halting problem is Turing undecidable, we have reached a contradiction. Therefore, the assumed Turing machine does not exist.
\end{proof}

\begin{remark}\label{RmkAlternativeGödelClass}
    We can imitate the construction of $\mathcal{H}_M$ for formal systems. Namely, with $F$ and $\varphi$ as in Subsection \ref{SbSctGoedelUndecidableNontrivialityFiniteness}, we can define, for $a\in c_c(\{0,1\})$, $$\tilde{g}_a(n)=\begin{cases} a(n) \quad&\textrm{if }\varphi(1),\ldots,\varphi(n)\textrm{ are consistent} \\ 0 &\textrm{else}\end{cases},$$
    and the function class $\tilde{\mathcal{G}}_F:= \{\tilde{g}_a\}_{a\in c_c(\{0,1\})}.$ Then, $\tilde{\mathcal{G}}_F= c_c(\{0,1\})$ if $F$ is consistent and $\tilde{\mathcal{G}}_F = \{0,1\}^{0,\ldots,K-1}$ for some finite $K\in\mathbb{N}$ if $F$ is inconsistent. Also, the mapping $F\mapsto\tilde{\mathcal{G}}_F$ is computable. We see that, for a suitable $F$, whether $\tilde{\mathcal{G}}_F$ is finite or equals $c_c(\{0,1\})$ is Gödel undecidable.
\end{remark}

\begin{remark}\label{RmkRice}
    One can also derive our Turing undecidability result using Rice's theorem. Informally, Rice's theorem states that any non-trivial semantic property of Turing machines is Turing undecidable \cite{rice1953classes}. 
    Combined with our computable mapping $M\mapsto\mathcal{H}_M$, Rice's Theorem yields uncomputability of the indicator function of the non-trivial index set $\{ M\in\mathbb{N} ~|~ \lvert \mathcal{H}_M\rvert <\infty \}$, thus providing an alternative (possibly more elegant but arguably less concrete) proof for Corollary~\ref{CrlTuringUndecidabilityFiniteness} once Corollary~\ref{CrlTuringMappingComputable} is established. 
    Here, an index set is a set of codes of partial computable functions such that whenever it contains a code of a partial computable function, then it contains all codes for that function.
    Note that this reasoning can in principle be generalized as follows: If $\IN\ni M\mapsto\mathcal{F}_M\subseteq\{0,1\}^{\IN}$ is a computable mapping from (codes of) Turing machines to computable hypothesis classes such that there exist Turing machines $M_f,M_i\in\IN$ with $\lvert \mathcal{F}_{M_f}\rvert<\infty$ and $\lvert \mathcal{F}_{M_i}\rvert=\infty$ and such that whenever $M_1,M_2\in\IN$ are two codes for the same Turing machine, then $\lvert \mathcal{F}_{M_1}\rvert<\infty$ holds if and only if $\lvert\mathcal{F}_{M_2}\rvert<\infty$, then Rice's Theorem implies that the indicator function of $\{ M\in\mathbb{N} ~|~ \lvert\mathcal{F}_M\rvert<\infty \}$ is uncomputable. 
\end{remark}

\begin{remark}\label{RmkConnectionTuringGoedel}
    There is a standard way, explained, e.g, in Section $2$ of \cite{Poonen.2014}, of deriving a Gödel undecidability result from a Turing undecidability result. For suitable (recursively enumerable, sufficiently expressive, and sound) formal systems, this allows us to derive from Corollary \ref{CrlTuringUndecidabilityFiniteness} the existence of a function class for which finiteness is Gödel undecidable.
    
    The advantage of our reasoning in Subsection \ref{SbSctGoedelUndecidableNontrivialityFiniteness}:
    Even without a soundness assumption, Corollary \ref{CrlGoedelUndecidabilityNonTriviality} provides us with a concrete example of a function class for which finiteness (and, as we show later, also finiteness of the VC-dimension) is Gödel undecidable. In that sense, the relationship between Corollary \ref{CrlGoedelUndecidabilityNonTriviality} and the Gödel undecidability result just derived from Corollary \ref{CrlTuringUndecidabilityFiniteness} is similar to the relationship between Gödel's second and (Rosser's \cite{Rosser.1936} strengthening of) the first incompleteness theorem.
    
    In fact, starting from a Turing undecidability result, one can derive a Gödel undecidabiltiy result akin to the second incompleteness theorem, compare the essays \cite{oberhoff2019incompleteness, cubitt2021spectralincompleteness}. In our case, starting from Corollary \ref{CrlTuringUndecidabilityFiniteness}, given a recursively enumerable formal system $F$, one can explicitly describe a Turing machine $M$, depending on $F$, such that neither $\lvert \mathcal{H}_M\rvert<\infty$ nor $\lvert\mathcal{H}_M\rvert=\infty$ can be proved in $F$. This $\mathcal{H}_M$ is then a concrete function class for which finiteness of $\lvert\mathcal{H}_M\rvert$ is Gödel undecidable in $F$ and thus gives a result comparable to Corollary \ref{CrlGoedelUndecidabilityNonTriviality}. We have presented our results on independence of the axioms of a formal system and on uncomputability separately, so that these parts of the paper can be read independently from one another.
\end{remark}

\section{Undecidability of Learnability}\label{SctUndecidableLearnability}

\subsection{Undecidability in PAC Binary Classification}\label{SbSctUndecidablePACClassification}

We now apply the results of Section \ref{SctUndecidableNontrivialityFiniteness} to PAC binary classification. First, aiming towards Gödel undecidability, observe the following immediate consequence of Proposition \ref{PrpConsistencyCollapseFctClass} and Theorem \ref{ThmInconsistencyHighComplexity}:

\begin{corollary}\label{CrlEquivalenceConsistencyFinitenessVCDim}
    $F$ is consistent iff $\VC (\mathcal{G}_F)<\infty$.
\end{corollary}

If the formal system $F$ is capable of expressing both the class $\mathcal{G}_F$ and the finiteness of its VC-dimension, we can again invoke Gödel's second incompleteness theorem. 

\begin{corollary}\label{CrlGoedelUndecidabilityFinitenessVCDim}
Assume that $F$ is a recursively enumerable and consistent formal system that contains elementary arithmetic, that can reason about $\mathcal{G}_F$, and such that infinitely many different theorems can be proved in $F$. Then $\VC (\mathcal{G}_F)<\infty$, but the finiteness of $\VC (\mathcal{G}_F)$ cannot be proved in $F$.
\end{corollary}

Note that the statement of Corollary \ref{CrlGoedelUndecidabilityFinitenessVCDim} is sensitive to the representation of $\mathcal{G}_F$. Namely, if $F$ is consistent, then $\mathcal{G}_F=\{0\}$ by Proposition \ref{PrpConsistencyCollapseFctClass}, and the VC-dimension of the trivial class $\{0\}$ is easily shown to be finite. However, with $\mathcal{G}_F$ computably defined via a function as explained in Subsection \ref{SbSctGoedelUndecidableNontrivialityFiniteness}, $F$ can neither prove that $\mathcal{G}_F$ is trivial nor that its VC-dimension is finite.

Turing undecidability of finiteness of the VC-dimension follows similarly easily. Namely, Proposition \ref{PrpTuringClassAlternativeForm} implies:

\begin{corollary}\label{CrlHaltingFiniteVCDim}
    Let $M$ be a Turing machine. The binary-valued function class $\mathcal{H}_{M}$ satisfies
    \begin{align*}
        \operatorname{VCdim}(\mathcal{H}_{M})
        =\begin{cases} K\quad&\textrm{if }M\textrm{ halts after exactly }K\textrm{ steps on the empty input}\\
        \infty &\textrm{else}\end{cases}.
    \end{align*}
    In particular, $\VC (\mathcal{H}_{M})<\infty$ if and only if $M$ halts on the empty input.
\end{corollary}

Combining this with the computability of $M\mapsto \mathcal{H}_M$ established in Corollary \ref{CrlTuringMappingComputable}, a reduction to the halting problem yields:

\begin{corollary}\label{CrlTuringUndecidabilityFinitenessVCDim}
    There is no Turing machine that, upon input of the code of an arbitrary computable binary-valued function class coming from $\{\mathcal{H}_M\}_{M}$, decides whether that class has finite VC-dimension. In other words, finiteness of the VC-dimension is Turing undecidable.
\end{corollary}

As finiteness of the VC-dimension is equivalent to learnability in PAC binary classification (recall Theorem \ref{ThmFundamentalTheoremBinaryClassification}), Corollaries \ref{CrlGoedelUndecidabilityFinitenessVCDim} and \ref{CrlTuringUndecidabilityFinitenessVCDim} can be rephrased as establishing the Gödel and Turing undecidability of learnability in this setting, respectively.

\begin{remark}\label{RmkFinitenessVCDimVSFinitenessClass}
    Since there are infinite function classes with finite VC-dimension, undecidability of the (in-)finiteness of function classes does not imply undecidability of the (in-)finiteness of the VC-dimension, or vice versa.
    However, as our constructed function classes $\mathcal{G}_F$ ($\mathcal{H}_M$) have a finite VC-dimension if and only if they are finite, which is the case if and only if $F$ is consistent ($M$ halts on the empty input), they simultaneously allow us to prove both undecidability of the (in-)finiteness and undecidability of the (in-)finiteness of the VC-dimension. The same applies to undecidability of (in-)finiteness of the complexity measures considered in the next two subsections.
\end{remark}

\subsection{Undecidability in Teacher-Learner Interactions}\label{SbSctUndecidableTeaching}

Next, we show undecidability results for teacher-learner interactions. Before coming to the teaching dimension itself, we discuss a different problem in teaching. Namely, we ask whether, given a function that can be taught to a learner by a teacher using finitely many examples, we can always prove that this is the case. The answer will turn out to be no, in general.

Take $F$ and $\varphi$ as in Subsection \ref{SbSctGoedelUndecidableNontrivialityFiniteness}. We consider the class of threshold functions on $\IN$ and allow for the possibility of a ``threshold at infinity.'' That is, we consider
\begin{align*}
\mathcal{F}_{\textrm{step}} := \{\IN\ni n\mapsto \operatorname{sgn}(n-k)~|~ k\in\IN\}\cup\{0\},
\end{align*}
where we use the convention $\operatorname{sgn}(x) = \begin{cases} 1\quad&\textrm{if }x\geq 0\\ 0&\textrm{if }x<0 \end{cases}$. Note that $\mathcal{F}_{\textrm{step}}$ consists of computable functions and is a computable class in the sense introduced in Definition \ref{DffComputableFunctionClass}.

We consider the function
\begin{align*}
    f_F:\IN\mapsto \{0,1\},\quad f_F(n) = \begin{cases} 0\quad &\textrm{if }\varphi(1),\ldots,\varphi(n)\textrm{ are consistent}\\ 1&\textrm{else}\end{cases}.
\end{align*}
Note that the mapping $F\mapsto f_F$, where we think of $F$ as given via the code of the corresponding $\varphi$, is computable. Clearly, $f_F\in\mathcal{F}_{\textrm{step}}$ for any formal system $F$. Therefore we can study whether $f_F$ admits a finite teaching set in the class $\mathcal{F}_{\textrm{step}}$.

\begin{proposition}\label{PrpEquivalenceInconsistencyFinitenessTSet}
$f_F$ admits a finite teaching set in $\mathcal{F}_{\textrm{step}}$ iff $F$ is inconsistent.
\end{proposition}
\begin{proof}
If $F$ is inconsistent, then there exists $k\in\IN$ such that $f_F(n)=\operatorname{sgn}(n-k)$ for all $n\in\IN$. So $f_F$ is the only element of $\mathcal{F}_{\textrm{step}}$ that is consistent with the training data set $\{ (k-1,0),(k,1)\}$. Thus, we have found a teaching set of size $2$ for $f_F$.

If $F$ is consistent, then $f_F\equiv 0$ is the zero-function. So any finite training data set consistent with $f_F$ is of the form $\{(n_i,0) \}_{i=1}^N$ for $n_i\in\IN$, $1\leq i\leq N$, $N\in\IN$. But also the function $\IN\ni n\mapsto \operatorname{sgn}(n-k^\ast)$ with $k^\ast = \max_{1\leq i\leq N} n_i +1$ is an element of $\mathcal{F}_{\textrm{step}}$ that is consistent with such a training data set. So $f_F$ cannot be uniquely identified in $\mathcal{F}_{\textrm{step}}$ by a finite training data set. That is, $f_F$ does not have a teaching set of finite size.
\end{proof}

We see that the teaching dimension of $\mathcal{F}_{\textrm{step}}$ is infinite. 
The formal system $F$ determines which element of $\mathcal{F}_{\textrm{step}}$ we consider, and Proposition \ref{PrpEquivalenceInconsistencyFinitenessTSet} states that, if $F$ is consistent, this ``filters out'' precisely the one concept in $\mathcal{F}_{\textrm{step}}$ that does not have a finite teaching set.

If $F$ is capable of expressing the function $f_F$, the class $\mathcal{F}_{\textrm{step}}$, and the (non-)existence of finite teaching sets, we are again in the position to apply Gödel's second incompleteness theorem.

\begin{corollary}\label{CrlGoedelUndecidabilityFiniteTeachingSet}
Assume that $F$ is a recursively enumerable and consistent formal system that contains elementary arithmetic and that can reason about $f_F$ and $\mathcal{F}_{\textrm{step}}$. The function $f_F$ defined above does not have a finite teaching set in $\mathcal{F}_{\textrm{step}}$, but this statement is not provable in $F$.
\end{corollary}

\begin{remark}\label{RmkTuringUndecidabilityFiniteTeachingSet}
We can use a similar construction to establish an analogous Turing undecidability result. Namely, given a Turing machine $M$, we can define 
\begin{align*}
f_M:\IN\mapsto \{0,1\},\quad f_M(n) = \begin{cases} 0\quad &\textrm{if } M \textrm{ does not halt after }\leq n\textrm{ steps on the empty input}\\ 1&\textrm{else}\end{cases}.
\end{align*}
$f_M$ admits a finite teaching set in $\mathcal{F}_{\textrm{step}}$ if and only if $M$ halts on the empty input. Hence, as the mapping $M\mapsto f_M$ is computable, we conclude that there cannot be a general-purpose algorithm that, upon input of a computable function class and a function in that class, decides whether the function admits a finite teaching set in the class. 
\end{remark}

Next, we study $\mathcal{G}_F$ from the perspective of the teaching dimension. For the purpose of this discussion, $F$, $\varphi$ and $E^2$ are again as in Subsection \ref{SbSctGoedelUndecidableNontrivialityFiniteness}. Our first observation is that also finiteness of the teaching dimension of $\mathcal{G}_F$ can be related to consistency of underlying formal system.

\begin{proposition}\label{PrpEquivalenceConsistencyFinitenessTDim}
$F$ is consistent iff $\operatorname{Tdim}(\mathcal{G}_F)<\infty.$
\end{proposition}
\begin{proof}
The ingredients to this proof are similar to those used in proving Corollary \ref{ThmInconsistencyHighComplexity}. See Appendix \ref{AppendixProofs} for details.
\end{proof}

The proof of Proposition \ref{PrpEquivalenceConsistencyFinitenessTDim} shows that, if $F$ is inconsistent, then in fact no element of $\mathcal{G}_F$ has a finite teaching set. This is different from our previous result, where a single function in $\mathcal{F}_{\textrm{step}}$ required teaching sets of infinite size and whether this function was the one characterized by $F$ depended on (in-)consistency.

Again, if $F$ can reason about $\mathcal{G}_F$ and the finiteness of its teaching dimension, we can combine Proposition \ref{PrpEquivalenceConsistencyFinitenessTDim} with Gödel's second incompleteness theorem:

\begin{corollary}\label{CrlGoedelUndecidabilityFinitenessTDim}
Assume that $F$ is a recursively enumerable and consistent formal system that contains elementary arithmetic, that can reason about $\mathcal{G}_F$, and such that infinitely many theorems can be proved in $F$. Then $\mathcal{G}_F$ has finite teaching dimension, but this statement cannot be proved in $F$.
\end{corollary}

Thus, we have shown that also the teaching dimension captures the contrast between the ``collapse'' of $\mathcal{G}_F$ in the consistent case and the ``richness'' of $\mathcal{G}_F$ in the inconsistent case. Therefore, finiteness of the teaching dimension is also Gödel undecidable.

\begin{remark}
If we leave aside questions of computability, we could also consider the following construction: Take $\tilde{\mathcal{F}}_{\textrm{step}}\subseteq\{0,1\}^\IN$ to be the class of proper step functions and consider the class $\{f_F\}\cup \tilde{\mathcal{F}}_{\textrm{step}}$. This class has finite teaching dimension if and only if $F$ is inconsistent.
\end{remark}

We can also view $\mathcal{H}_M$ through the lens of teacher-learner interactions. As before, the first step in our approach consists in relating whether the underlying Turing machine $M$ halts to whether the teaching dimension of $\mathcal{H}_M$ is finite.
\begin{proposition}\label{PrpEquivalenceHaltingFinitenessTDim}
Let $M$ be a Turing machine. The binary-valued function class $\mathcal{H}_M$ satisfies
\begin{align*}
\operatorname{Tdim}(\mathcal{H}_{\textrm{M}})
=\begin{cases} K\quad&\textrm{if }M\textrm{ halts after exactly }K\in\IN\textrm{ steps on the empty input}\\
\infty &\textrm{else}\end{cases}.
\end{align*}
\end{proposition}
\begin{proof}
This follows from Proposition \ref{PrpTuringClassAlternativeForm}.
\end{proof}

As we already know from Corollary \ref{CrlTuringMappingComputable} that $\mathcal{H}_M$ can be computed from the underlying Turing machine $M$, we can again reduce to the halting problem and obtain

\begin{corollary}\label{CrlTuringUndecidabilityFinitenessTDim}
There is no Turing machine that, upon input of the code of an arbitrary computable binary-valued function class coming from $\{\mathcal{H}_M\}_{M}$, decides whether that class has finite teaching dimension. In other words, finiteness of the teaching dimension is Turing undecidable.
\end{corollary}

\subsection{Undecidability in Online Learning}\label{SbSctUndecidableOnlineLearning}

As a final demonstration of the applicability of our constructions to different learning models, we show that universal and uniform online learning are undecidable in the by now familiar two senses.
We have seen in Theorem \ref{ThmUniformOnlineLearningFiniteLDim} that uniform online learnability is equivalent to the Littlestone dimension being finite. Therefore, we again start by relating consistency of the formal system $F$ underlying $\mathcal{G}_F$ to finiteness of the Littlestone dimension of $\mathcal{G}_F$.

\begin{proposition}\label{PrpEquivalenceConsistencyFinitenessLDim}
    $F$ is consistent iff $\operatorname{Ldim}(\mathcal{G}_F)<\infty.$
\end{proposition}
\begin{proof}
    This follows from our results on the VC-dimension and the inequality $\VC\leq\operatorname{Ldim}$. See Appendix \ref{AppendixProofs} for details.
\end{proof}

The Gödel undecidability of uniform online learnability now follows as in Subsection \ref{SbSctUndecidablePACClassification}:

\begin{corollary}\label{CrlGoedelUndecidabilityFinitenessLDim}
    Assume that $F$ is a recursively enumerable and consistent formal system that contains elementary arithmetic, that can reason about $\mathcal{G}_F$, and such that infinitely many theorems can be proved in $F$. Then $\mathcal{G}_F$ has finite Littlestone dimension, but this statement cannot be proved in $F$.
\end{corollary}

We also have a teaching dimension analogue of Proposition \ref{PrpEquivalenceHaltingFinitenessTDim}:

\begin{proposition}\label{PrpHaltingFiniteLDim}
Let $M$ be a Turing machine. The binary-valued function class $\mathcal{H}_{M}$ satisfies
\begin{align*}
\operatorname{Ldim}(\mathcal{H}_{M})
=\begin{cases} K\quad&\textrm{if }M\textrm{ halts after exactly }K\textrm{ steps on the empty input}\\
\infty &\textrm{else}\end{cases}.
\end{align*}
\end{proposition}
\begin{proof}
    This follows from our results on the VC-dimension and the inequalities $\VC(\mathcal{H}_{M})\leq\operatorname{Ldim}(\mathcal{H}_{M})\leq \log_2\lvert\mathcal{H}_{M}\rvert$. See Appendix \ref{AppendixProofs} for details.
\end{proof}

Relying again on the computability of $M\mapsto\mathcal{H}_M$, we can reduce to the halting problem:

\begin{corollary}\label{CrlTuringUndecidabilityFinitenessLDim}
    There is no Turing machine that, upon input of the code of an arbitrary computable binary-valued function class coming from $\{\mathcal{H}_M\}_{M}$, decides whether that class has finite Littlestone dimension. In other words, finiteness of the Littlestone dimension is Turing undecidable.
\end{corollary}

We have similar results for universal online learning.
Because of Theorem \ref{ThmUniversalOnlineLearningNoInfiniteLittlestoneTrees}, we first establish an equivalence between the formal system $F$ underlying the class $\mathcal{G}_F$ being consistent and $\mathcal{G}_F$ having no infinite Littlestone tree.

\begin{proposition}\label{PrpEquivalenceConsistencyNoInfiniteLittlestoneTree}
    $F$ is consistent iff $\mathcal{G}_F$ does not have an infinite Littlestone tree.
\end{proposition}
\begin{proof}
    This follows from Proposition \ref{PrpConsistencyCollapseFctClass} and Theorem \ref{ThmInconsistencyHighComplexity}. See Appendix \ref{AppendixProofs} for details.
\end{proof}

With this observation, the same reasoning, using Gödel's second incompleteness theorem, as in Subsection \ref{SbSctUndecidablePACClassification} yields:

\begin{corollary}\label{CrlGoedelUndecidabilityNoInfiniteLittlestoneTree}
    Assume that $F$ is a recursively enumerable and consistent formal system that contains elementary arithmetic and that can reason about $\mathcal{G}_F$, and such that infinitely many theorems can be proved in $F$. Then $\mathcal{G}_F$ does not have an infinite Littlestone tree, but this statement cannot be proved in $F$.
\end{corollary}

Again, we can also focus our attention on $\mathcal{H}_M$. Here, we obtain:

\begin{proposition}\label{PrpHaltingNoInfiniteLittlestoneTree}
    Let $M$ be a Turing machine. The binary-valued function class $\mathcal{H}_{M}$ has an infinite Littlestone tree iff $M$ does not halt on the empty input.
\end{proposition}
\begin{proof}
    This follows from Proposition \ref{PrpConsistencyCollapseFctClass} and Theorem \ref{ThmInconsistencyHighComplexity}. See Appendix \ref{AppendixProofs} for details.
\end{proof}

\begin{corollary}\label{CrlTuringUndecidabilityInfiniteLittlestoneTree}
There is no Turing machine that, upon input of the code of an arbitrary computable binary-valued function class coming from $\{\mathcal{H}_M\}_{M}$, decides whether that class has an infinite Littlestone tree. In other words, the existence of infinite Littlestone trees is Turing undecidable.
\end{corollary}

As before, because of the computability of $M\mapsto\mathcal{H}_M$, through Theorem \ref{ThmUniversalOnlineLearningNoInfiniteLittlestoneTrees}, this implies that universal online learnability is both Gödel and Turing undecidable.

\section{Conclusion and Open Questions}\label{SctConclusion}

In this work, we have shown that in the standard model of binary classification, in two models of online learning, and in a basic model describing teacher-learner interactions, it is in general undecidable whether the learning task can be completed, given only a computable description of the hypothesis class for the task. We have established this for two different meanings of ``undecidable,'' the first being ``true, but not provable in a formal system'' and the second being ``not computable.'' In both cases, our results follow by providing computable constructions that allow for a reduction of the problem of deciding finiteness of the complexity measure for the respective learning task to the prototypic undecidable problem, i.e., to proving consistency of a formal system or to deciding whether a Turing machine halts.

It was already known, due to \cite{BenDavid.2019}, that learnability can be independent of the axioms of ZFC. We have proved a similar-in-spirit result for the arguably most influential learning model, the PAC model of binary classification. By discussing our proof strategy also for a teacher-learner model and for online learning, we have demonstrated that it is not specific to the PAC setting. Moreover, learnability can be undecidable also in other formal systems and in the terminology of computer science. A crucial feature of our constructions, especially for establishing undecidability in the latter sense, is that we are only dealing with computable objects. This is to be contrasted with \cite{BenDavid.2019}, where the continuum is used. In particular, the arguments of \cite{BenDavid.2019} do not give uncomputability results. Note that, since our constructions are computable, instead of PAC learnability, equivalently we could have considered computable PAC learnability, because of \cite[Theorem $10$]{Agarwal.2020}, when restricting our attention to the realizable scenario. 

We emphasize that our undecidability results crucially depend on the representation of the hypothesis classes. That is, our results rely on a hypothetical decision algorithm having access to a computable hypothesis class via its code. We have shown that there is no general way of deciding, given the code of a hypothesis class from a uniformly computable family, whether that class is learnable. As demonstrated in our proofs, this happens because it is impossible to determine from such a representation whether the hypothesis class is simple or highly complex. Thus, when evaluating a learning task and the associated function class, the representation plays an important role. Formal and algorithmic representations, such as the one via a function of parameters and inputs put forward in this work, can be insufficient for determining the complexity and thus the learnability of a hypothesis class of interest. 

We hope that our work adds to the ongoing research aiming towards a better understanding of the theoretical prospects and limits of machine learning. Our results indicate that potential problems for applications of machine learning do not only arise on the level of algorithmic design, which in itself is an extremely challenging task. Rather, already when faced with a learning task, we encounter a fundamental difficulty: It is in general not possible to decide whether that task is, from the of sample complexity, leaving questions of computational complexity aside, learnable, i.e., in principle amenable to a solution via machine learning.

From a more practical perspective, our results can be interpreted as follows: When faced with a learning task, one can usually choose which hypothesis class to use. This choice will be guided by different considerations, such as prior knowledge about the problem, potential issues for optimization, and questions of learnability. In particular, one usually chooses a class that is known to be learnable. Thereby, the ``library'' of candidate function classes is restricted to those whose learnability has already been established. Our results say that there is no generic way of enlarging this library: Every time one faces a learning problem for which all classes from the current library perform poorly, identifying a new suitable candidate class, even leaving questions of optimization aside, presents a new challenge because of learnability alone. One possibility of facing this challenge in practice may be to use additional information beyond the description of the candidate class. For example, for PAC binary classification, \cite{Vapnik.1994} proposed a general method for measuring an effective version of the VC-dimension based on experimental estimates for the expected worst-case generalization error on training data of different sizes. 

Finally, we mention some questions raised by our work:
\begin{itemize}
\item We have approached learnability through criteria based on complexity measures of the function class under consideration. Can (un-)decidability be established for learnability via algorithmic properties, e.g., stability \cite{bousquet2002stability} or compression-based schemes \cite{littlestone1986relating}?
\item For our PAC learning scenario, we require a sample complexity bound that is uniform over the function class. Can our results be extended to non-uniform learning models in which the sample size is allowed to depend on the function to be learned?
\item As discussed in Remark \ref{RmkFinitenessVCDimVSFinitenessClass}, it would be interesting to see whether the finiteness of the VC-, teaching, and Littlestone dimension remain undecidable also when restricting the potential inputs to codes of infinite function classes.
\item As observed in \cite{Bousquet.2020}, universal online learning is closely connected to Gale-Stewart games. Do undecidability results in one of these two scenarios translate to the respectively other one? For example, can we recover the undecidability results for Gale-Stewart games due to \cite{Rabin.1958} and \cite{Jones.1982} from our results on universal online learning? Or can we use these works to gain further insight into undecidability in online learning?
\end{itemize}

\newpage
\section*{Acknowledgements}
I thank Michael M.~Wolf for stimulating discussions on questions of undecidability and for suggesting the reasoning leading to Corollary \ref{CrlGoedelUndecidabilityFiniteTeachingSet}. 
I also thank the anonymous reviewers and the meta-reviewer from COLT 2021, the anonymous reviewers from ITCS 2022, the anonymous reviewers from COLT 2022, the anonymous reviewers from JMLR, and the anonymous reviewers at the International Journal for Approximate Reasoning for their feedback.
Furthermore, I thank Artem Chernikov, Asaf Karagila, Aryeh Kontorovich, Vladimir Pestov, and Roi Weiss for pointing the measure-theoretic subtleties around Theorem \ref{ThmFundamentalTheoremBinaryClassification} out to me. 
Moreover, I thank Aryeh Kontorovich for bringing the references \cite{papadimitriou1996limited, mossel2002complexity, hanneke2021universal} to my attention.
Finally, I thank Tom Sterkenburg for an illuminating discussion leading to a clarified statement of Remark \ref{RmkRice}.

Support from the TopMath Graduate Center of TUM the Graduate School at the Technische Universität München, Germany, from the TopMath Program at the Elite Network of Bavaria, from the German Academic Scholarship Foundation (Studienstiftung des deutschen Volkes), from the BMWK (PlanQK), and from the DAAD PRIME Fellowship program is gratefully acknowledged.

\setcounter{secnumdepth}{0}
\defbibheading{head}{\section{References}}
\printbibliography[heading=head]

\newpage
\appendix
\section*{Appendix}
\setcounter{secnumdepth}{2}

\section{Proofs}\label{AppendixProofs}

\begin{proof}[Proof of Proposition \ref{PrpExFalsoQuodlibet}]
As $F$ is inconsistent, there exists a theorem $p$ such that both $p$ and $\neg p$ can be proved in $F$. As $p$ can be proved in $F$, we have $q\to p$ (which is our notation for ``$q$ implies $p$''). By negation we then have $\neg p \to \neg q$. As $\neg p$ can be proved in $F$, also $\neg q$ can be proved in $F$. If we now exchange $q$ by $\neg q$ in the above reasoning, we see that also $q$ can be proved in $F$.
\end{proof}

\begin{proof}[Proof of Lemma \ref{LmmRecursiveEnumerabilityCompactlySupportedBinarySequences}]
We define $C:\IN\times\IN\to\{ 0,1\}$ via
\begin{align*}
C(m,n) = \begin{cases} n^{\textrm{th}}\textrm{ bit in the binary representation of } m \quad &\textrm{ if } m>0\textrm{ and } 2^n\leq m\\ 0 &\textrm{ else } \end{cases}.
\end{align*}
As exponentiation and finding the binary representation of a natural number can be done in a primitive recursive manner, the function $C$ is defined in terms of primitive recursive functions and a case distinction via a primitive recursive predicate Thus, $C$ is itself primitive recursive.

Clearly, the function $n\mapsto C(m,n)$ has finite support and is thus an element of $c_c(\{ 0,1\})$. Conversely, if $a\in c_c(\{0,1\})$, then there exists $K\in\IN$ such that $a(n)=0$ for all $n>K$. Hence, if we take $m_a$ to be the natural number with binary representation $a(0)\ldots a(K)$, then $a(n)=C(m_a,n)$ for all $n\in\IN$.
\end{proof}

\begin{proof}[Proof of Proposition \ref{PrpEquivalenceConsistencyFinitenessTDim}]
If $F$ is consistent, $\mathcal{G}_F=\{ 0\}$ by Proposition \ref{PrpConsistencyCollapseFctClass} and the claim is trivial. If $F$ is inconsistent, then uniquely identifying a function $g_a\in\mathcal{G}_F$ requires one to uniquely identify the subsequence $(a_{k_l})_{l\in\IN}$ of $a\in c_c(\{0,1\})$ chosen such that $k_{l+1}>k_l$ and such that $\varphi(E^2_1(k))=\neg\varphi(E^2_2(k))$ iff $k=k_l$ for some $l\in\IN$. As $\varphi(E^2_1(k))=\neg\varphi(E^2_2(k))$ is satisfied for infinitely many $k\in\IN$ (see Proposition \ref{PrpExFalsoQuodlibet}) and the size of the support of an element of $c_c(\{0,1\})$ can be arbitrarily large, any training data set that uniquely identifies $g_a$ has to consist of infinitely many labelled examples.
\end{proof}

\begin{proof}[Proof of Proposition \ref{PrpEquivalenceConsistencyFinitenessLDim}]
    If $F$ is consistent, $\mathcal{G}_F=\{0\}$ and clearly $\operatorname{Ldim}(\mathcal{G}_F)=0$. If $F$ is inconsistent, we can use the well known inequality $\VC\leq \operatorname{Ldim}$ together with Corollary \ref{CrlEquivalenceConsistencyFinitenessVCDim} to obtain $\operatorname{Ldim}(\mathcal{G}_F)=\infty$.
\end{proof}

\begin{proof}[Proof of Proposition \ref{PrpHaltingFiniteLDim}]
    This follows quite directly the well known fact that, for any function class $\mathcal{G}\subset\{0,1\}^\mathcal{X}$, $\VC(\mathcal{G})\leq \operatorname{Ldim}(\mathcal{G})\leq \log_2 \lvert\mathcal{G}\rvert$.

    Namely, if $M$ halts on the empty input, these two inequalities, due to Corollary \ref{CrlHaltingFiniteVCDim} and Proposition \ref{PrpTuringClassAlternativeForm}, become $K\leq \operatorname{Ldim}(\mathcal{G})\leq \log_2 \lvert\mathcal{H}\rvert= K$. And if $M$ does not halt on the empty input, the lower bound via the VC-dimension, together with Proposition \ref{PrpTuringClassAlternativeForm}, implies $\operatorname{Ldim}(\mathcal{G})=\infty$. 
\end{proof}

\begin{proof}[Proof of Proposition \ref{PrpEquivalenceConsistencyNoInfiniteLittlestoneTree}]
If $F$ is consistent, then $\operatorname{Ldim}(\mathcal{G}_F)<\infty$ by Proposition \ref{PrpEquivalenceConsistencyFinitenessLDim}. In particular, $\mathcal{G}_F$ does not have an infinite Littlestone tree.

If $F$ is inconsistent, then, by Theorem \ref{ThmInconsistencyHighComplexity}, there exists a sequence $(n_k)_{k=1}^{\infty}\subset\mathbb{N}$ such that $\{n_1,\ldots,n_N\}$ is shattered by $\mathcal{G}_F$ for every $N\in\mathbb{N}_{\geq 1}$. Therefore, we obtain an infinite Littlestone tree of $\mathcal{H}_M$ by labelling every node in the $k^{\textrm{th}}$ layer by $n_{k+1}$, for $k\in\mathbb{N}_0$. 
\end{proof}

\begin{proof}[Proof of Proposition \ref{PrpHaltingNoInfiniteLittlestoneTree}]
    If $M$ halts on the empty input, then $\operatorname{Ldim}(\mathcal{H}_M)<\infty$ by Proposition \ref{PrpHaltingFiniteLDim}. In particular, $\mathcal{H}_M$ does not have an infinite Littlestone tree.

    If $M$ does not halt on the empty input, then $\{0,\ldots,N\}$ is shattered by $\mathcal{H}_M$ for every $N\in\mathbb{N}$, as we have seen in the proof of Corollary \ref{CrlHaltingFiniteVCDim}. Therefore, we obtain an infinite Littlestone tree of $\mathcal{H}_M$ by labelling every node in the $k^{\textrm{th}}$ layer by $k$, for $k\in\mathbb{N}_0$. 
\end{proof}

\section{Gödel and Incompleteness of Formal Systems}\label{SctGoedelPreliminaries}

Here, we compile standard notions connected to formal systems which appear in the main body of the paper. However, some notions will only be introduced informally and the interested reader is referred to other sources for the formal definitions.

We denote by $\IN$ the natural numbers including $0$. We call a function $f:\IN^k\to\IN$ \emph{primitive recursive} if it can be built from the zero function, the successor function, and the coordinate projection functions via composition and primitive recursion. From a modern perspective, the primitive recursive functions are those that can be implemented using basic arithmetic as well as IF THEN ELSE, AND, OR, NOT, $=$, $>$, and FOR loops. WHILE loops are not allowed here.

Next, we recall, albeit only informally, the notion of a formal system.
\begin{definition}[Formal systems - Informal]\label{DffFormalSystems}
A \emph{formal system} $F$ consists of a finite alphabet of symbols, a language of statements that can be well-formed from the alphabet, a distinguished set of statements called \emph{axioms}, and rules for how to derive/prove new theorems from these axioms.\\
A formal system $F$ is called \emph{consistent} if there is no well-formed statement such that both it and its negation can be proved in $F$. Otherwise, we call $F$ \emph{inconsistent}.
\end{definition}

We will be interested in a particular kind of formal systems in which the provable theorems, i.e., the statements that can be deduced from the axioms according to the derivation rules, can be recursively enumerated. To make this assumption more rigorous, we first recall 

\begin{definition}[Gödel numbering - Informal]\label{DffGödelNumbering}
A \emph{Gödel numbering} for a formal system $F$ is an injective function that maps each symbol in the alphabet and each well-formed statement to an element of $\IN$.
\end{definition}

For our purposes, it does not matter which Gödel numbering is used. We only use that Gödel numberings exist for which ``translating'' between a string of Gödel numbers of symbols describing a statement and the actual Gödel number of that statement can be done primitive recursively in both directions. Gödel's original construction has this property. From now on, we fix such a Gödel numbering. This allows us to identify statements in a formal system with elements of $\IN$ and ``manipulations'' of statements with primitive recursive maps between natural numbers. Both of these identifications will sometimes be implicit throughout the paper.

From this perspective, we can describe the type of formal systems used in this work.

\begin{definition}[Recursively enumerable formal systems]\label{DffRecursivelyEnumerableFormalSystems}
A formal system $F$ is called \emph{recursively enumerable} (or \emph{effectively axiomatized}) if there exists a primitive recursive function $\varphi:\IN\to\IN$ such that $\{\varphi(n)~|~n\in\IN\}$ is exactly the set of all Gödel numbers in a fixed Gödel numbering of statements that can be proved in $F$.
\end{definition}

Given such a primitive recursive enumeration $\varphi$ of provable theorems, we will sometimes abuse notation and take $\varphi(n)$ to denote both a theorem and its Gödel number. The exact meaning, if not made explicit, will be clear from the context.

Gödel's second incompleteness theorem provides, for any recursively enumerable and consistent formal system that contains elementary arithmetic, an explicit statement that is true but cannot be proved in that formal system. 

\begin{theorem}[Gödel's second incompleteness theorem \cite{Godel.1931}]\label{ThmGoedelIncompleteness}
Assume that $F$ is a recursively enumerable and consistent formal system that contains elementary arithmetic. Then the consistency of $F$ is not provable in $F$.
\end{theorem}

We call a statement that is true but not provable in a formal system $F$ \emph{Gödel undecidable} in $F$. This is not standard terminology, we merely use it to shorten some formulations.

For a more formal presentation of these and other notions from mathematical logic, the reader is referred to textbooks such as \cite{Barwise.1993, Kleene.2002, Enderton.2013}.

\section{Turing and Uncomputability}\label{SctTuringPreliminaries}

This section recalls standard definitions and results related to Turing machines and computability. Again, sometimes we give only an informal presentation and refer to textbooks for details.

In \cite{Turing.1937}, Turing introduced what are now known as a Turing machines. We do not give a formal definition, but instead describe the workings of a Turing machine informally. For a more rigorous presentation, see, e.g., \cite{Davis.1982, Soare.2016}.

\begin{definition}[Turing machines - Informal]\label{DffTMs}
A Turing machine $M$ consists of
\begin{itemize}
\item a $1$-dimensional tape with infinitely many cells extending in both directions, each of which contains a symbol from a finite alphabet $\Sigma$, 
\item a head that can read and write symbols in a single cell and move to the left or to the right by one cell, 
\item a finite set of states $Q$ containing an \emph{initial state} and a \emph{halting state}, 
\item and an instruction function $I:\Sigma\times Q\to\Sigma\times\{L,R\}\times Q$ describing the write-, move- and state-update-behaviour of $M$ upon reading a given symbol while in a given state.
\end{itemize}
\end{definition}
The two distinguished states are the initial state, in which the Turing machine begins any of its computations, and the halting state, that causes the Turing machine to halt when it is reached.

According to the Church-Turing thesis, which could be considered a ``law of nature'' for the world of computing, everything that can be reasonably considered computable is computable by a Turing machine. Hence, we take Turing machines as our model for defining computability.

\begin{definition}[(Turing) Computable functions]\label{DffTuringComputableFcts}
A partial function $f:\IN^k\to\IN$, for $k\in\IN_{\geq 1}$, is \emph{(Turing) computable} if there exists a finite-state Turing machine $M$ such that, whenever we run $M$ on a tape with an encoding of $x\in\operatorname{dom}(f)$ written on it, $M$ eventually halts with the tape containing an encoding of $f(x)$, and whenever we run $M$ on a tape with an encoding of $x\not\in\operatorname{dom}(f)$ written on it, $M$ does not halt.
\end{definition}

One possible choice of encoding is the unary encoding. That is, $x\in\IN$ is represented by $x+1$ consecutive ones on the tape. The remaining tape is left blank. An element of $\IN^k$ can be represented by $k$ blocks of unary encodings of the components, separated by single zeros. 

It is useful to note at this point that any primitive recursive function is computable. However, there are computable functions that are not primitive recursive.

We call a decision problem whose corresponding function, mapping instances of the problem to a binary ``yes-or-no'' output, is not computable \emph{Turing undecidable}. The prototypic example of a Turing undecidable decision problem is the \emph{halting problem}, i.e., the problem of deciding whether a given Turing machine halts on the empty input. Already \cite{Turing.1937} observed that this cannot be achieved in a computable way.

We will also use a notion of computability of function classes. 

\begin{definition}\label{DffComputableFunctionClass}
We say that a class $\mathcal{G}\subseteq\IN^\IN$ is computable if there exists a total computable function $G:\IN\times\IN\to\IN$ such that $\mathcal{G}=\{n\mapsto G(m,n)~|~m\in\IN\}$.
\end{definition}

We recall one last fact related to Turing machines. Namely, there exist \emph{universal Turing machines} capable of simulating any Turing machine \cite{Turing.1937}. From now on, for each $k\in\mathbb{N}_{\geq 1}$, we fix such a universal Turing machine understood as a partial computable function $\mathbb{M}:\IN^{k+1}\to\IN$. Then, for any Turing machine $M$ and corresponding partial computable function $f_M:\IN^k\to\IN$, there exists a natural number, also denoted by $M$, such that $f_M(x)=\mathbb{M}(M,x)$ for every $x\in\IN$. The natural number $M$ is called the \emph{code} of the Turing machine $M$ with respect to $\mathbb{M}$. This allows us to think of Turing machines, or, equivalently, computable functions, as input when representing them by their code with respect to our fixed universal Turing machine.

Definition \ref{DffComputableFunctionClass} can be interpreted as requiring $\mathcal{G}$ to admit a function $G$ that, upon input of a ``parameter'' $m$, determining an element of $\mathcal{G}$, and of an input value $n$, computably evaluates the $m$th element of $\mathcal{G}$ on $n$. If such a function exists, then we can identify $\mathcal{G}$ with $G$ and use the code of $G$ as a description of $\mathcal{G}$ that may serve as input to a decision algorithm. 

\end{document}